\pdfpagewidth=\paperwidth
\pdfpageheight=\paperheight
\documentclass[12pt, letterpaper]{article}
\pdfoutput=1

\DeclareFixedFont{\MyTitleFont}{OT1}{ptm}{m}{n}{22pt}
\DeclareFixedFont{\MyAuthorFont}{OT1}{ptm}{m}{n}{13pt}
\DeclareFixedFont{\MyAbstractTitleFont}{OT1}{ptm}{m}{it}{12pt}
\DeclareFixedFont{\MyAbstractFont}{OT1}{ptm}{m}{it}{11pt}
\DeclareFixedFont{\MySubtitleFont}{OT1}{ptm}{m}{n}{15pt}
\DeclareFixedFont{\MySubSubtitleFont}{OT1}{ptm}{m}{n}{13pt}
\DeclareFixedFont{\MySubSubSubtitleFont}{OT1}{ptm}{m}{n}{11pt}
\DeclareFixedFont{\MyTextFont}{OT1}{ptm}{m}{n}{11pt}

\usepackage{geometry}
\emergencystretch = \maxdimen
\usepackage{titling}

\title{\MyTitleFont Sensitivity Analysis of Continuous-Time Systems based on Power Spectral Density\vspace{0em}}

\author{\MyAuthorFont Neng Wan$^1$, Dapeng Li$^2$, and Naira Hovakimyan$^1$}
\date{}

\usepackage{titlesec}

\usepackage[marginal]{footmisc}
\usepackage{abstract}

\usepackage[sort]{cite}

\usepackage{indentfirst}
\usepackage{scalefnt}
\usepackage{amsthm}
\usepackage{amsfonts}
\usepackage{amsmath}
\usepackage{txfonts}
\usepackage{booktabs}
\usepackage{bm}
\usepackage{mathptmx}
\usepackage{subfigure}
\usepackage{enumerate}
\usepackage[mathcal]{euscript}
\usepackage{float}

\usepackage[pdftex]{color,graphicx}
\usepackage[dvipsnames]{xcolor}

\usepackage[pdftex, bookmarks]{hyperref}
\hypersetup{colorlinks = true, citecolor = red, linkcolor = blue, urlcolor = black}

\newtheorem{theorem}{Theorem}
\newtheorem{remark}{Remark}
\newtheorem{assumption}{Assumption}
\newtheorem{lemma}{Lemma}

\newtheorem{definition}{Definition}
\newtheorem{corollary}[theorem]{Corollary}

\allowdisplaybreaks

\begin{document}

\maketitle

\footnotetext[0]{$^{1}$Neng Wan and Naira Hovakimyan are with the Department of Mechanical Science and Engineering, University of Illinois at Urbana-Champaign, Urbana, IL 61801, USA. {\tt\small \{nengwan2, nhovakim\}@illinois.edu}.}
\footnotetext[0]{$^{2}$Dapeng Li is a Principal Scientist with the JD.com Silicon Valley Research Center, Santa Clara, CA 95054, USA. {\tt\small dapeng.li@jd.com}.} 
\vspace{-4em}

\begin{abstract}
{Bode integrals of sensitivity and sensitivity-like functions along with complementary sensitivity and complementary sensitivity-like functions are conventionally used for describing performance limitations of a feedback control system. In this paper, we show that in the case when the disturbance is a wide sense stationary process the (complementary) sensitivity Bode integral and the (complementary) sensitivity-like Bode integral are identical. A lower bound of the continuous-time complementary sensitivity-like Bode integral is also derived and examined with the linearized flight-path angle tracking control problem of an F-16 aircraft.}
\end{abstract}
\vspace{-0.5em}

\titleformat*{\section}{\centering\MySubtitleFont}
\titlespacing*{\section}{0em}{1.25em}{1.25em}[0em]

\section{Introduction}\label{sec1}

The last two decades have witnessed a tremendous progress in communication technologies and their use in feedback control systems. A great deal of attention has been given to understanding the fundamental limitations of closed-loop systems in the presence of communication channels \cite{Martin_TAC_2007, Martin_TAC_2008, Okano_Auto_2009, Ishii_SCL_2011, Fang_TAC_2017, Li_TAC_2013}. The main contribution of these papers was to derive performance limitations of stochastic nonlinear systems in the presence of limited information. While \cite{Martin_TAC_2007, Martin_TAC_2008, Okano_Auto_2009, Ishii_SCL_2011, Fang_TAC_2017} looked into discrete-time systems and investigated the Bode-like integrals using Kolmogorov's entropy-rate equality~\cite{Cover_2012}, the results in \cite{Li_TAC_2013} provided an extension to continuous-time systems by resorting to mutual information rates. In these papers, the notion of the sensitivity-like function was introduced to derive Bode-like integrals and corresponding lower bounds, which can be considered as a generalization of the classical result of Bode integrals for linear time-invariant (LTI) deterministic systems~\cite{Bode_1945}. The classical result in~\cite{Bode_1945} states that  for open-loop stable transfer functions the Bode integral equals zero, while for unstable open-loop transfer functions it is lower bounded by the sum of unstable poles of the open-loop transfer function~\cite{Freudenberg_1985, Freudenberg_1987}. Similar to the sensitivity function in a LTI system, the complementary sensitivity function is also used for robustness and performance analysis of closed-loop systems~\cite{Middleton_1991}. We notice that the result on the complementary sensitivity Bode integral was once hindered by the unboundedness of the integrand in high frequencies~\cite{Sung_IJC_1989}. This issue was later overcome in \cite{Middleton_1991} by adopting a weighted Bode integral of the complementary sensitivity function, proven to be lower bounded by the sum of the reciprocals of non-minimum phase zeros. Seminal results on this topic were reported also in~\cite{Zhou_1998, Seron_2012}.

Performance limitations of stochastic systems in the presence of limited information were analyzed through sensitivity-like function $S(\omega)$ in~\cite{Martin_TAC_2007, Martin_TAC_2008, Fang_TAC_2017, Li_TAC_2013} and the complementary sensitivity-like function $T(\omega)$ in~\cite{Okano_Auto_2009, Ishii_SCL_2011}. Taking an information-theoretic approach was the key to get  Bode integrals extended to  stochastic nonlinear systems. Unlike the frequency-domain approach, which explicitly depends on the input-output relationship of the feedback  systems (transfer function), the focus of the information-theoretic approach is on the signals. The lower bound for sensitivity-like Bode integral for  continuous-time systems was first put forward in~\cite{Li_TAC_2013}:
\begin{equation}\label{eq10}
\dfrac{1}{2\pi}\int_{-\infty}^{\infty} \log |S(\omega)|  d\omega \geq \sum_{\lambda \in \mathcal{UP}} p_i.
\end{equation}
This result can be applied to systems with nonlinear controllers, which is an improvement upon the prior results based on the frequency-domain approach~\cite{Bode_1945, Freudenberg_1985, Freudenberg_1987, Sung_IJC_1988, Zhou_1998, Middleton_1991, Sung_IJC_1989, Seron_2012}. However, to the best of authors' knowledge, a lower bound for the complementary sensitivity-like Bode integral for continuous-time systems has not been derived yet. The unboundedness of the integrand in high frequencies as stated in~\cite{Sung_IJC_1989} and the challenge in representing the weighted Bode-like integral with information-theoretic tools similar to~\cite{Middleton_1991} have been the main obstacles on this path.

In this paper, we provide a partial answer to the question: \textit{What is the relationship between Bode integrals of the (complementary) sensitivity function and the (complementary) sensitivity-like function?} We answer this question for the continuous-time linear feedback system with a wide sense stationary input, while some partial answers on discrete-time systems can be found in~\cite{Martin_TAC_2008, Ishii_SCL_2011}.  We notice that while Kolmogorov's entropy-rate equality has been used for discrete-time systems in~\cite{Martin_TAC_2007, Martin_TAC_2008, Okano_Auto_2009, Ishii_SCL_2011, Fang_TAC_2017} to obtain a lower bound for the sensitivity Bode-like integral, a seminal result on mutual information rates from \cite[p.~181]{Pinsker_1964} was  used in~\cite{Li_TAC_2013} to obtain a similar bound for continuous-time systems. In this paper, we resort to power spectral density (PSD) to analyze the sensitivity and the complementary sensitivity of continuous-time systems. With the convenience brought by this new tool, we first time find a lower bound and an information-theoretic representation for the complementary sensitivity Bode-like integral. The sensitivity properties of an F-16 aircraft in the flight-path angle tracking problem are analyzed.

The paper is organized as follows: \hyperref[sec2]{Section II} introduces the preliminaries on Bode integrals and information theory. \hyperref[sec3]{Section III} investigates the relationship between the sensitivity and the sensitivity-like Bode integrals. \hyperref[sec4]{Section IV} investigates the complementary sensitivity and the complementary sensitivity-like Bode integrals and proposes a lower bound for the latter. \hyperref[sec5]{Section V} presents a numerical example. \hyperref[sec6]{Section VI} draws the conclusion.

\section{Preliminaries}\label{sec2}
Consider a continuous-time feedback configuration $\mathcal{P}$ depicted in~\hyperref[fig1]{Figure 1},
\begin{figure}[htpb]\label{fig1}
	\centering
	\vspace{-3em}
	\includegraphics[width=0.48\textwidth]{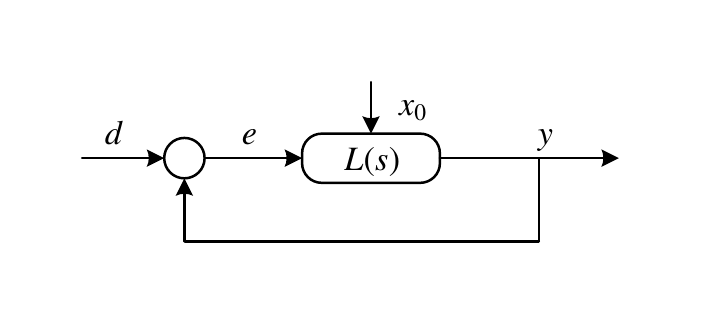}\vspace{-3em}\\
	\caption{Continuous-time feedback control system.}
\end{figure}

\noindent where $d(t)\in\mathbb{R}$ is the disturbance input, $y(t) \in \mathbb{R}$ is the output, $e(t) = d(t) - y(t)$ is the error signal, $x(t) \in \mathbb{R}^n$ is the state, and $L(s)$ denotes the open-loop transfer function from $e(t)$ to $y(t)$
\begin{equation}
L(s) = L(j\omega) = \int_{0}^{\infty} l(t) \cdot {\rm e}^{-j\omega t} dt,
\end{equation}
with $l(t)$ being the impulse response of the system. In a deterministic setting, the initial condition $x_0$ in the configuration of~\hyperref[fig1]{Figure 1} is assumed zero. In a stochastic setting, one assumes that  the differential entropy of the initial condition is  finite~\cite{Martin_TAC_2007, Martin_TAC_2008, Fang_TAC_2017, Li_TAC_2013}.  Further discussion on these two different types of initial conditions is available in~\cite{Ishii_SCL_2011}. Let the open-loop transfer function $L(s)$ in~\hyperref[fig1]{Figure~1} be
\begin{equation}\label{eq2}
L(s) = \dfrac{Y(s)}{E(s)} = c \cdot \dfrac{\prod_{j = 1}^m (s - z_j)}{\prod_{i=1}^n(s - p_i)},
\end{equation}
where $m \leq n$, and $c>0$. Inspired by~\cite{Middleton_1991}, consider the following frequency transformation
\begin{equation}\label{eq30}
\tilde{s} = j\tilde{\omega} = (j\omega)^{-1} = s^{-1},
\end{equation}
where $\tilde{\omega} = - \omega^{-1}$. Applying~\eqref{eq30} to transfer function~\eqref{eq2}, the system with following transfer function $\tilde{L}(\tilde{s})$ is defined as the auxiliary system:
\begin{equation}
\tilde{L}(\tilde{s}) = \dfrac{\tilde{Y}(\tilde{s})}{\tilde{E}(\tilde{s})} = c \cdot \dfrac{\tilde{s}^n \cdot \prod_{j=1}^m (1 - \tilde{s} \cdot z_j)}{\tilde{s}^m \cdot \prod_{i=1}^{n} (1 - \tilde{s} \cdot p_i)} = L(s),
\end{equation}

\noindent which can be depicted by the diagram in~\hyperref[fig2]{Figure 2}.
\begin{figure}[htpb]
	\centering\label{fig2}
	\vspace{-3em}
	\includegraphics[width=0.48\textwidth]{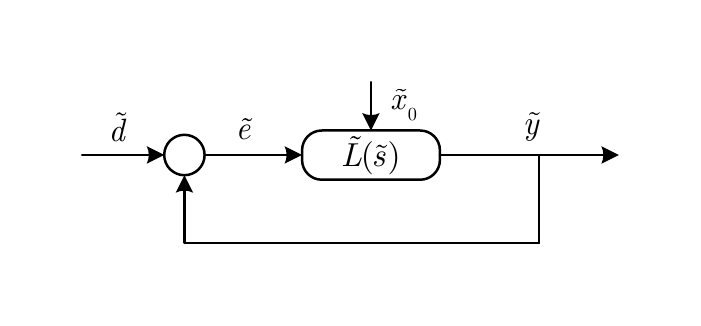}\vspace{-3em}\\
	\caption{Auxiliary system.}
\end{figure}

\noindent The Laplace transforms of the signals in the auxiliary system and the signals in the original system satisfy
\begin{equation}\label{eq66}
\tilde{D}(\tilde{s}) = \tilde{D}(s^{-1}) = D(\tilde{s}^{-1}) = D(s),
\end{equation}
which will also hold if $d$ is replaced by $e$ or $y$. It is worth noting that although the auxiliary system $\tilde{L}(\tilde{s})$ may not be proper, no intermediate result will be derived from this auxiliary system. The inverse system $\tilde{L}^{-1}(\tilde{s})$ is defined by swapping the input $\tilde{e}$ and the output $\tilde{y}$ of the auxiliary system. The transfer function of this inverse system then becomes:
\begin{equation}
{\tilde{L}}^{-1}(\tilde{s}) = \dfrac{\tilde{E}(\tilde{s})}{\tilde{Y}(\tilde{s})} = \dfrac{1}{c  }\cdot \dfrac{1}{\tilde{s}^{n-m}} \cdot \dfrac{ \prod_{i=1}^{n} (1 - \tilde{s} \cdot p_i)}{ \prod_{j=1}^m (1 - \tilde{s} \cdot z_j)},
\end{equation}
\noindent which is illustrated in~\hyperref[fig3]{Figure 3}.
\begin{figure}[htpb]
	\centering\label{fig3}
	\vspace{-3em}
	\includegraphics[width=0.48\textwidth]{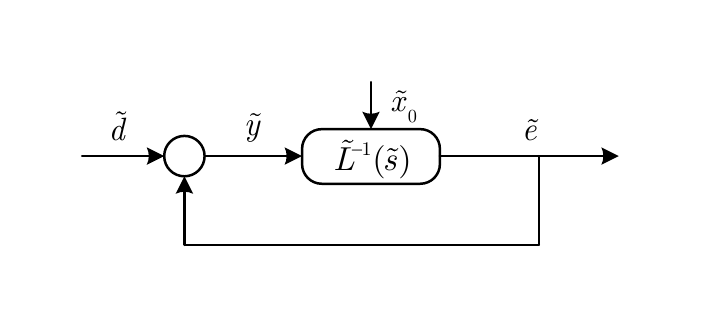}\vspace{-3.5em}\\
	\caption{Inverse of auxiliary system.}
\end{figure}

\noindent One can easily verify that if all the closed-loop poles of the original system is stable, the closed-loop poles of the inverse system will also be stable. To generalize the results of this paper to MIMO systems, interested readers can refer to~\cite{Ishii_SCL_2011, Fang_2017}. Before we continue to formulate the (complementary) sensitivity analysis problem, some basic definitions are given below following~\cite{Li_TAC_2013, Cover_2012}.

\begin{definition}[Wide Sense Stationary]
	A second order random process $\{x\}$ is called wide sense stationary, if
	\begin{equation}
	\begin{split}
	\mathbb{E}[x(t)] &= \mathbb{E}[x(t+v)],\\
	\textrm{Cov}[x(t), x(t+\tau)] &= \textrm{Cov}[x(v), x(v+\tau)],
	\end{split}
	\end{equation}
	where $\mathbb{E}$ denotes expectation.
\end{definition}

\begin{definition} \textbf{(Mutual Information \& Mutual Information Rate)}
	The mutual information between two continuous-time stochastic processes $x$ and $y$ is defined as
	\begin{equation}
	I(x;y) = \int_Y\int_X f(x,y) \log \dfrac{f(x,y)}{f(x)f(y)} dxdy,
	\end{equation}
	where $f(x,y)$ is the joint probability distribution function, and $f(x)$ and $f(y)$ are the marginal probability distribution functions. The mutual information rate is defined as
	\begin{equation}
	I_\infty(x; y) = \lim_{t \rightarrow \infty} \dfrac{I(x^t; y^t)}{t} .
	\end{equation}
\end{definition}

\begin{definition}[Class $\mathbb{F}$ Function\textnormal{; See \cite{Li_TAC_2013} or \cite[p.~182]{Pinsker_1964}}]  We define class $\mathbb{F}$  function in the following way:
	\begin{equation}
	\mathbb{F} = \{l:l(\omega) = p(\omega)(1-\varphi(\omega)), l(\omega) \in \mathbb{C}, \omega \in \mathbb{R} \},
	\end{equation}
	where $p(\cdot)$ is rational and $\varphi(\cdot)$ is a measurable function, such that $0 \leq \phi \leq 1$ for all $\omega \in \mathbb{R}$ and $\int_{\mathbb{R}}|\log(1-\varphi(\omega))| d\omega < \infty$.
\end{definition}

The sensitivity function $S(j\omega)$ of the feedback system in \hyperref[fig1]{Figure 1} is defined as the closed-loop transfer function from the disturbance input $d$ to the tracking error $e$:
\begin{equation}\label{eq3}
S(j\omega) = \dfrac{E(j\omega)}{D(j\omega)} = \dfrac{1}{1 + L(j\omega)} .
\end{equation}
The complementary sensitivity function $T(j\omega)$ is defined as the closed-loop transfer function from the disturbance input $d$ to the measurement output $y$:
\begin{equation}\label{eq4}
T(j\omega) = \dfrac{Y(j\omega)}{D(j\omega)}= \dfrac{L(j\omega)}{1+L(j\omega)}.
\end{equation}
The integrals of $S(j\omega)$ and $T(j\omega)$ over the whole frequency domain are referred to as Bode integrals and satisfy the following equalities~\cite{Freudenberg_1985,Freudenberg_1987,Seron_2012}:
\begin{equation}\label{eq5}
\scalefont{0.95}{\dfrac{1}{2\pi}\int_{-\infty}^{\infty} \log \left|\dfrac{S(j\omega)}{S(j\infty)}\right|  d\omega =  \lim_{s\rightarrow\infty} \dfrac{s[S(s) - S(\infty)]}{2 \cdot S(\infty)}+  \sum_{p_i \in \mathcal{UP}} p_i},
\end{equation}
\begin{equation}\label{eq6}
\dfrac{1}{2\pi}\int_{-\infty}^{\infty} \log \left|\dfrac{T(j\omega)}{T(0)}\right|  \dfrac{d\omega}{\omega^2} = \dfrac{1}{2T(0)} \lim_{s\rightarrow 0} \dfrac{d T(s)}{ds} + \sum_{z_i \in \mathcal{UZ}} \frac{1}{z_i},
\end{equation}
where $\mathcal{UP}$ and $\mathcal{UZ}$ respectively denote the set of unstable poles and the set of non-minimum phase zeros of the plant $\mathcal{P}$. Since~\eqref{eq5} and~\eqref{eq6} are derived in frequency domain using transfer functions, they cannot be applied to nonlinear systems.

Starting with~\cite{Martin_TAC_2007, Martin_TAC_2008}, information theoretic tools were leveraged to derive performance limitations and Bode-like results for nonlinear systems. Instead of considering the sensitivity function $S(j\omega)$, in~\cite{Martin_TAC_2008, Li_TAC_2013} sensitivity-like function $S(\omega)$ was introduced based on the properties of signals:
\begin{equation}\label{s_like}
S(\omega) = \sqrt{\dfrac{\phi_e(\omega)}{\phi_d(\omega)}},
\end{equation}
\noindent where $\phi_x(\omega)$ denotes the PSD of a stationary signal $x$:
\begin{equation}
\phi_x(\omega) = \int_{-\infty}^{\infty} r_x(\tau) \cdot  {\rm e}^{-j\omega\tau} d\tau,
\end{equation}
and $r_x(\tau) = r_{xx}(t + \tau,t)$ denotes the auto-covariance of the signal $x$ with
\begin{equation*}
r_{xy}(v,t) = \textrm{Cov}[x(v), y(t)] .
\end{equation*}
The complementary sensitivity-like function was defined  for discrete-time systems in~\cite{Okano_Auto_2009}. Following the same philosophy, the following definition of the complementary sensitivity-like function is adopted in this paper:
\begin{equation}\label{eq9}
T(\omega) = \sqrt{\dfrac{\phi_y(\omega)}{\phi_d(\omega)}}.
\end{equation}
As we mentioned previously, the lower bound for Bode integral of $T(\omega)$ in continuous-time systems has not been studied yet. In the following sections, we first discuss the relationship between the (complementary) sensitivity and the (complementary) sensitivity-like Bode integrals and then propose a lower bound for the Bode integral of $T(\omega)$. Some lemmas and assumptions that we adopt in this paper are listed next.
\begin{lemma}\label{lem1}
	(See \cite{Li_TAC_2013} or \cite[p.~181]{Pinsker_1964}) Suppose that two one-dimensional continuous-time processes $x$ and $y$ form a stationary Gaussian process $(x, y)$. Then
	\begin{equation}
	I_\infty(x,y) \geq -\dfrac{1}{4\pi}\int_{-\infty}^{\infty}\log\left(1 - \dfrac{|\phi_{xy}(\omega)|^2}{\phi_x(\omega)\phi_y(\omega)}\right) d\omega.
	\end{equation}
	The equality holds, if $\phi_x$ and $\phi_y$ belong to the class $\mathbb{F}$.
\end{lemma}

\begin{assumption}\label{ass1}
	The disturbance input $d(t)$ is a zero-mean wide sense stationary process.
\end{assumption}
\begin{remark}
	Compared with~\cite{Martin_TAC_2007, Fang_TAC_2017, Martin_TAC_2008, Okano_Auto_2009}, which assumed that $d$ is an asymptotically stationary process,  \hyperref[ass1]{Assumption~1} is relatively stringent. However, this assumption is commonly adopted among the results on continuous-time systems in terms of signals,~\cite{Astrom_2012, Li_TAC_2013}.
\end{remark}

\begin{assumption}\label{ass2}
	For the transfer function $L(s)$ the amount of zeros at $s = 0$ does not exceed the amount of poles at $s = 0$.
\end{assumption}

\begin{remark}
	We only adopt this assumption when establishing a lower bound for the complementary sensitivity-like Bode integral. This assumption ensures that the inverse system $\tilde{L}^{-1}(\tilde{s})$  is proper, \textit{e.g.} for a double integrator vehicle with first order actuator dynamics $L(s) = 1/[s^2\cdot (0.1s+1)]$ from~\cite{Seiler_TAC_2004}, we have $\tilde{L}^{-1} = (\tilde{s}+0.1)/\tilde{s}^3$. Similar assumption was adopted in~\cite{Middleton_2010}, when investigating the string instability (sensitivity) via a frequency-domain approach.
\end{remark}

\section{Sensitivity and Sensitivity-Like Functions}\label{sec3}
We first investigate the relationship between  Bode integrals of sensitivity function $S(j\omega)$ and sensitivity-like function $S(\omega)$ of the closed-loop configuration in~\hyperref[fig1]{Figure 1}. The following theorem states this relationship.
\begin{theorem}\label{thm1}
	\textit{When the disturbance input $d(t)$ is wide sense stationary,  Bode integrals of the sensitivity and the sensitivity-like functions satisfy}
	\begin{equation}\label{eq1_0}
	\dfrac{1}{2\pi}\int_{-\infty}^{\infty} \log {S}(\omega) d\omega = \dfrac{1}{2\pi}\int_{-\infty}^{\infty} \log |S(j\omega)|  d\omega.
	\end{equation}
\end{theorem}
\begin{proof}
	The proof  is given in \hyperref[appA]{Appendix A}.
\end{proof}

\section{Complementary Sensitivity and Sensitivity-Like Functions}\label{sec4}
The relationship between Bode integrals of complementary sensitivity function $T(j\omega)$ and the complementary sensitivity-like function $T(\omega)$ in~\hyperref[fig1]{Figure 1} is summarized in the following corollary.
\begin{corollary}\label{thm2}
	\textit{When the disturbance input $d(t)$ is wide sense stationary,  Bode integrals of complementary sensitivity function $T(j\omega)$ and complementary sensitivity-like function ${T}(\omega)$ satisfy}
	\begin{equation}\label{eq24}
	\frac{1}{2\pi} \int_{-\infty}^{\infty} \log{T}(\omega)  \frac{d\omega}{\omega^2} = \frac{1}{2\pi} \int_{-\infty}^{\infty} \log |T(j\omega)|  \dfrac{d\omega}{\omega^2}.
	\end{equation}
\end{corollary}
\begin{proof}
	The proof is given in \hyperref[appB]{Appendix B}.
\end{proof}

\vspace{0.5em}

From~\hyperref[thm2]{Corollary 2}, we know that Bode integrals of $T(j\omega)$ and $T(\omega)$ are equivalent, when the disturbance input is wide sense stationary. The following theorem gives a lower bound for the Bode integral of $T(\omega)$ in continuous-time setting.

\begin{theorem}\label{thm3}
	\textit{When the original system in~\hyperref[fig1]{Figure 1} is mean-square stable and the inverse frequency noise $\tilde{d}$ is wide sense stationary, one has:
		\begin{equation}
		I_\infty(\tilde{y}; \tilde{e}) - I_\infty(\tilde{d}; \tilde{e}) \geq \sum_{z_i \in \mathcal{UZ}} \dfrac{1}{z_i},
		\end{equation}
		where $\mathcal{UZ}$ is the set of unstable zeros of the plant $\mathcal{P}$, and $\tilde{e}$ and $\tilde{y}$ are the signals defined in the (inverse) auxiliary system. Moreover, when the disturbance input $\tilde{d}$ is Gaussian stationary, the complementary sensitivity-like Bode integral satisfies
		\begin{equation}\label{eq29}
		\frac{1}{2\pi} \int_{-\infty}^{\infty} \log {T}(\omega) \frac{d\omega}{\omega^2} \geq \sum_{z_i \in \mathcal{UZ}} \dfrac{1}{z_i}.
		\end{equation}
	}
\end{theorem}
\begin{proof}
	By the frequency transform~\eqref{eq30}, we can rewrite the complementary sensitivity-like Bode integral defined in~\eqref{eq24} as follows
	\begin{equation}\label{eq24_2}
	\begin{split}
	\frac{1}{2\pi} \int_{-\infty}^{\infty} \log {T}(\omega) \frac{d\omega}{\omega^2} & = \frac{1}{2\pi} \int_{-\infty}^{\infty} \log {T}(-\tilde{\omega}^{-1}) d \tilde{\omega}\\
	& = \frac{1}{2\pi} \int_{-\infty}^{\infty} \log {\tilde T}(\tilde{\omega}) d \tilde{\omega},\\
	\end{split}
	\end{equation}
	where by~\hyperref[thm2]{Corollary~2} the complementary sensitivity-like function of auxiliary system $\tilde{T}(\tilde{\omega})$ satisfies
	\begin{equation}\label{eq25}
	\tilde{T}(\tilde{\omega}) = \sqrt{\dfrac{\phi_{\tilde{y}}(\tilde{\omega})}{\phi_{\tilde{d}}(\tilde{\omega})}} = \sqrt{\dfrac{\phi_{{y}}(-\tilde{\omega}^{-1})}{\phi_{{d}}(-\tilde{\omega}^{-1})}} = T(\omega).
	\end{equation}
	Meanwhile, since the complementary sensitivity-like function of the auxiliary system is identical to the sensitivity-like function of the inverse system, our task becomes to seek  a lower bound for the sensitivity Bode-like integral for the inverse system shown in~\hyperref[fig3]{Figure 3}. Since the inverse frequency noise $\tilde{d}$ is a wide sense stationary process, applying Theorem 4.8 in~\cite{Li_TAC_2013} to the inverse system, we have
	\begin{equation}\label{eq_36}
	I_{\infty}(\tilde{y}; \tilde{e}) - I_{\infty}(\tilde{d}; \tilde{e}) \geq \sum_{z_i \in \mathcal{UZ}} \dfrac{1}{z_i}.
	\end{equation}
	When the disturbance $\tilde d$ is stationary Gaussian, according to~\eqref{eq25} and Theorem 4.8 in~\cite{Li_TAC_2013}, we have
	\begin{equation}\label{eq28}
	\begin{split}
	&\frac{1}{2\pi} \int_{-\infty}^{\infty} \log {T}(\omega) \frac{d\omega}{\omega^2} = \frac{1}{2\pi} \int_{-\infty}^{\infty} \log {\tilde T}(\tilde{\omega}) d \tilde{\omega} \\
	&\hspace{7em} = I_{\infty}(\tilde{y}; \tilde{e}) - I_{\infty}(\tilde{d}; \tilde{e})  \geq \sum_{z_i \in \mathcal{UZ}} \dfrac{1}{z_i}.
	\end{split}
	\end{equation}
	This completes the proof. \hfill \qed
\end{proof}

\begin{remark}
	Since $\log T(\omega) = \log |T(j\omega)|$ tends to infinity as $\omega \rightarrow \infty$, similar to~\eqref{eq6}, we define the Bode-like integral of $T(\omega)$ with a weighting factor $1 / \omega^2$ in~\eqref{eq29}. 
	We note that this weighting factor  induces some restrictions when analyzing the complementary sensitivity via  information-theoretic approach, such as the requirement of stationary Gaussian condition on the inverse frequency signal.
\end{remark}

\begin{remark}
	When the disturbance ${\tilde d}$ is Gaussian stationary and the initial condition ${\tilde x}_0$ is Gaussian, by \hyperref[lem1]{Lemma 1} we can express the mutual information rate $I_\infty(\tilde{y}, \tilde{e})$ in terms of the density functions of $e$ and $y$:
	\begin{equation}
	\begin{split}
	I_\infty(\tilde{y}, \tilde{e}) &= -\dfrac{1}{4\pi}\int_{-\infty}^{\infty}\log\left(1 - \dfrac{|\phi_{\tilde{y}\tilde{e}}(\tilde\omega)|^2}{\phi_{\tilde{y}}(\tilde\omega)\phi_{\tilde{e}}(\tilde\omega)}\right) d\tilde{\omega}\\
	& = -\dfrac{1}{4\pi} \int_{-\infty}^{\infty}\log\left(1 - \dfrac{|\phi_{ye}(\omega)|^2}{\phi_y(\omega)\phi_e(\omega)}\right) \dfrac{d\omega}{\omega^2}.
	\end{split}
	\end{equation}
	The expression of $I_\infty(\tilde{d}; \tilde{e})$ can be readily implied.
\end{remark}

\section{An Illustrative Example}\label{sec5}

With the lower bound of the complementary sensitivity Bode-like integral given in~\hyperref[thm3]{Theorem 3}, we now investigate the control trade-offs in an aircraft flight-path angle tracking problem. Considering an F-16 aircraft with $\textrm{Mach} = 0.7$ and altitude $h = 10,000 \textrm{\ ft}$, the linearized longitudinal dynamics can be described by the following state-space model~\cite{Shkolnikov_JGCD_2001}.
\begin{equation*}
\begin{split}
A & = \left[\begin{matrix}
-11.707 & 0 & -75.666 \\
0 & 11.141 & -79.908 \\
0.723 & 0.907 & -1.844\\
\end{matrix}\right], \quad
B = \left[\begin{matrix}
0 \\ 0\\ 0.117
\end{matrix}\right], \\
C & = \left[\begin{matrix}
0, \ 0, \ 1
\end{matrix}\right],
\end{split}
\end{equation*}
where the input is elevator deflection $\delta_e(t)$, and the output is flight-path angle $\gamma(t)$. With zero initial condition, the longitudinal dynamics in state-space form can be equivalently described by the following transfer function
\begin{equation*}
G(s) = \dfrac{0.117 \cdot (s + 11.71)(s - 11.14)}{(s + 2.979)(s - 1.051)(s+0.4826)},
\end{equation*}
which contains a non-minimum phase zero at $s = 11.14$ and an unstable pole at $s = 1.051$. Consider the following two PID controllers with different sets of parameters:
\begin{equation*}
\begin{split}
C_1(s) &= -0.4 - 0.06 \cdot \dfrac{1}{s} - 1 \cdot \dfrac{100}{1 + 100 \cdot 1 / s}, \\
C_2(s) &= 2 \cdot C_1(s),\\
\end{split}
\end{equation*}
where $100 / (1 + 100 / s)$ is an approximation of the derivative term in PID controller, and the open-loop transfer functions $L_1(s) = G(s)C_1(s)$ and $L_2(s) = G(s)C_2(s)$.

With the plant transfer function $G(s)$ and control mapping $C_1(s)$, we first verify the lower bounds of Bode-like integrals. By~\hyperref[lem1]{Lemma~1}, we can compute the Bode-like integral in~\eqref{eq29} with the complementary sensitivity function defined by $L_1(s)$, which gives
\begin{equation*}
\dfrac{1}{2\pi}\int_{-\infty}^{\infty} \log T(\omega) \ \dfrac{d\omega}{\omega^2} = 0.915 \geq 8.977 \times 10^{-2} = \sum_{z_i \in \mathcal{UZ}}\dfrac{1}{z_i}.
\end{equation*}
The sensitivity Bode-like integral can also be computed as
\begin{equation*}
\dfrac{1}{2\pi}\int_{-\infty}^{\infty} \log S(\omega) d\omega = 6.925  \geq 1.051 = \sum_{p_i \in \mathcal{UP}} p_i.
\end{equation*}
\begin{remark}
	Although both the sensitivity and complementary sensitivity Bode-like integrals are bounded in this example, for arbitrary causal transfer functions $L(s)$ that are closed-loop stable, these two Bode-like integrals are not guaranteed to be bounded. A comprehensive discussion on the boundedness of sensitivity Bode integral subject to the different conditions of the open-loop transfer functions $L(s)$ is available in~\cite{Wu_TAC_1992}.
\end{remark}

With the linearized longitudinal dynamics $G(s)$ and controller mappings $C_1(s)$ and $C_2(s)$, by~\hyperref[thm1]{Theorem~1} and~\hyperref[thm2]{Corollary~2}, the magnitudes of complementary sensitivity-like functions and sensitivity-like functions are given in~\hyperref[fig4]{Figure~4}, in which the solid lines denote the data with $C_1(s)$ and the dashed lines represent the data with $C_2(s)$. Subject to disturbance $d(t)$, the complementary sensitivity-like and sensitivity-like functions shown in~\hyperref[fig4]{Figure~4} tell that control mapping $C_1(s)$ performs better in disturbance mitigation in higher frequencies ($\omega > 5$ $\textrm{rad}\cdot \textrm{s}^{-1}$), while control mapping $C_2(s)$ performs better when attenuating the disturbance of lower frequencies ($\omega < 5$ $\textrm{rad}\cdot \textrm{s}^{-1}$), which can be explained by inequalities~\eqref{eq10} and~\eqref{eq29}, since the area below the solid line should equal to the area below the dashed line when the control mappings do not contain any unstable pole and non-minimum phase zero. This phenomenon is also known as the water-bed effect~\cite{Doyle_1992}. 

\begin{figure*}[t]
	\centering
	\includegraphics[width=\textwidth]{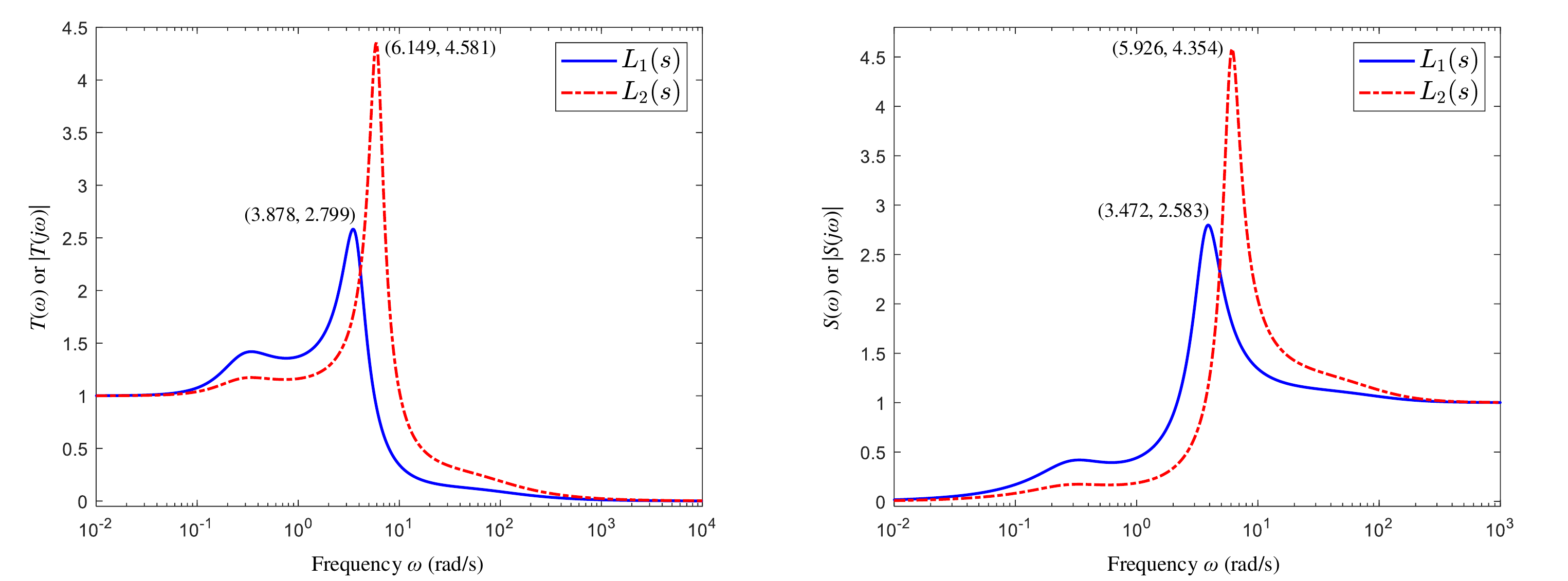} \vspace{-1.5em}
	\caption{Complementary sensitivity-like functions (left) and sensitivity-like integrals (right).}\label{fig4}
\end{figure*}

\section{Conclusions}\label{sec6}
We discussed the relationship between Bode integrals of (complementary) sensitivity functions and (complementary) sensitivity-like functions. A lower bound for the continuous-time complementary sensitivity Bode-like integral was derived based on the power spectral densities of signals. The lower bound was later examined with the linearized flight-path angle tracking control problem of an F-16 aircraft. Future discussions may include relaxing distribution condition on the disturbance signal and generalizing these results to  nonlinear systems.

\section*{Acknowledgment}
This work was partially supported by AFOSR and NSF. The authors specially acknowledge the staff and readers on arXiv.org.

\section*{Appendix}\label{app}

\subsection{Proof of Theorem 1}\label{appA}
Since $d(t) = e(t) + y(t)$, the density function $\phi_d(\omega)$ in~\eqref{s_like} satisfies
\begin{align}\label{eq1_2}
\phi_d(\omega)  &=  \int_{-\infty}^{\infty} r_d(\tau) \cdot {\rm e}^{-j\omega\tau} d\tau \nonumber\\
&\hspace{-5pt} = \int_{-\infty}^{\infty} \left[r_e(\tau) + r_{ey}(\tau) + r_{ye}(-\tau) + r_y(\tau)\right] {\rm e}^{-j\omega\tau} d\tau\nonumber\\
&\hspace{-5pt} = \phi_e(\omega) + \phi_{ey}(\omega) + \phi_{ye}(\omega) + \phi_y(\omega)
\end{align}
Letting $\tau = v - t$, and noticing that $y(t) = \int_{0}^{\infty} l(v')e(t - v')dv'$,	 subject to \hyperref[ass1]{Assumption 1},  the covariances $r_e, r_{ey}, r_{ye}$ and $r_y$, in~\eqref{eq1_2} satisfy
\begin{subequations}
	\begin{align}
	r_e(v,t) &= \textrm{Cov}[e(t+v-t), e(t)] \nonumber \\
	&\hspace{5em} =\textrm{Cov}[e(t+\tau), e(t)]= r_e(\tau)\\
	r_{ey}(v,t) & = \textrm{Cov}[e(v), \int_{0}^{\infty}l(v')e(t - v') dv']\nonumber\\
	& = \int_{0}^{\infty}l(v') r_e(v' + \tau ) dv'\nonumber\\
	& = r_{ey}(\tau)\nonumber\\
	r_{ye}(v,t) &= \textrm{Cov}[\int_{0}^{\infty}l(v')e(v-v')dv', e(t)]\nonumber\\
	& = \int_{0}^{\infty} l(v')r_e(-v' + \tau) dv' \\
	& = r_{ye}(\tau)\nonumber\\
	&\hspace{-4.25em}r_y(v,t)  = \textrm{Cov}[\int_{0}^{\infty}l(v')e(v-v')dv', \int_{0}^{\infty}l(t')e(t-t')dt']\nonumber\\
	& \hspace{-1.25em}= \int_{0}^{\infty}\int_{0}^{\infty}l(v')l(t') \cdot r_e(\tau - v' + t') dv' dt'\nonumber\\
	&\hspace{-1.25em} = r_y(\tau)
	\end{align}
\end{subequations}
\noindent Hence the spectral density functions $\phi_{ey}, \phi_{ye}$, and $\phi_y$, in~\eqref{eq1_2} satisfy
\begin{subequations}\label{eq20}
	\begin{align}
	\phi_{ey}(\omega) & = \frac{1}{2\pi} \int_{-\infty}^{\infty} r_{ey}(\tau) \cdot {\rm e}^{-j\omega\tau}d\tau \nonumber\\
	& \hspace{-1.25em} = \frac{1}{2\pi} \int_{0}^{\infty}  {\rm e}^{j\omega v'}l(v') \int_{-\infty}^{\infty} {\rm e}^{-j\omega(\tau+v')}  r_e(\tau + v') d\tau dv'\nonumber\\
	& \hspace{-1.25em} = L(-j\omega) \phi_e(\omega)\\
	\phi_{ye}(\omega) &= \frac{1}{2\pi} \int_{-\infty}^{\infty} r_{ye}(\tau) \cdot {\rm e}^{-j\omega\tau}d\tau\nonumber\\
	& \hspace{-1.25em} = \frac{1}{2\pi} \int_{0}^{\infty}  {\rm e}^{-j\omega v'}l(v')  \int_{-\infty}^{\infty} {\rm e}^{-j\omega(\tau-v')} r_e(\tau - v') d\tau dv'\nonumber\\
	& \hspace{-1.25em} = L(j\omega)\phi_e(\omega)\\
	\phi_y(\omega) & = \dfrac{1}{2\pi} \int_{-\infty}^{\infty} r_y(\tau) \cdot {\rm e}^{-j\omega\tau} d\tau\nonumber\\
	& = \int_{0}^{\infty} {\rm e}^{j\omega t'}l(t') \int_{0}^{\infty} {\rm e}^{-j\omega s'}l(s') \cdot \\
	& \hspace{3em}\int_{-\infty}^{\infty} {\rm e}^{-j\omega(\tau-s'+t')}r_e(\tau - s' + t') d\tau ds' dt'\nonumber\\
	& = L(-j\omega)L(j\omega)\phi_e(\omega)\nonumber
	\end{align}
\end{subequations}
Substituting~\eqref{eq1_2} and~\eqref{eq20} into the sensitivity-like function $S(\omega)$ defined in~\eqref{s_like}, we can rewrite the sensitivity-like function as follows
\begin{equation}
{S}(\omega) = \sqrt{\dfrac{\phi_e(\omega)}{[1+L(-j\omega)] \cdot [1+L(j\omega)] \cdot \phi_e(\omega)}}.
\end{equation}
When $\phi_e(\omega) \not\equiv 0$, we have	
\begin{equation}
{S}(\omega)  = \sqrt{S(-j\omega) \cdot S(j\omega)}.
\end{equation}
Since $S(-j\omega) = \bar{S}(j\omega)$, where $\bar{S}(j\omega)$ is the complex conjugate of ${S}(j\omega)$, the equality~\eqref{eq1_0} in~\hyperref[thm1]{Theorem 1} can be retrieved from

\begin{align}
\dfrac{1}{2\pi}\int_{-\infty}^{\infty} \log {S}(\omega) \ d\omega & = \dfrac{1}{4\pi} \int_{-\infty}^{\infty} \log \left[S(-j\omega) \cdot S(j\omega)\right] d \omega\nonumber\\
& =  \dfrac{1}{2\pi} \int_{-\infty}^{\infty} \log | S(j\omega) | d\omega .
\end{align}
This completes the proof. \hfill \qed

\subsection{Proof of Corollary 2}\label{appB}

Substituting~\eqref{eq1_2} and~\eqref{eq20} into the complementary sensitivity-like function $T(\omega)$ defined in~\eqref{eq9}, we can then rewrite $T(\omega)$ as follows
\begin{equation}
{T}(\omega)  = \sqrt{\dfrac{L(-j\omega)L(j\omega) \cdot \phi_e}{[1+L(-j\omega)]\cdot[1 + L(j\omega)] \cdot \phi_e}}
\end{equation}
When $\phi_e(\omega) \not\equiv 0$, it follows that
\begin{equation}
{T}(\omega) = \sqrt{T(-j\omega) \cdot T(j\omega)}
\end{equation}
Since $T(-j\omega) = \bar{T}(j\omega)$, where $\bar{T}(j\omega)$ is the complex conjugate of ${T}(j\omega)$, the equality~\eqref{eq24} in~\hyperref[thm2]{Corollary 2} can be retrieved from
\begin{align}
\frac{1}{2\pi} \int_{-\infty}^{\infty} \log {T}(\omega) \frac{d\omega}{\omega^2} & = \frac{1}{4\pi} \int_{-\infty}^{\infty} \log [T(-j\omega) \cdot T(j\omega)] \dfrac{d\omega}{\omega^2}\nonumber \\
& =  \frac{1}{2\pi} \int_{-\infty}^{\infty} \log |T(j\omega)| \dfrac{d\omega}{\omega^2}
\end{align}

\noindent This completes the proof.\hfill \qed

\end{document}